\documentclass[12pt]{article}
\usepackage{amsmath}
\usepackage{amssymb}
\usepackage{amsthm}

\usepackage{graphicx}
\usepackage{natbib}
\usepackage[margin=1.0in]{geometry}
\usepackage{authblk}
\usepackage{enumerate}
 \usepackage{algorithm, setspace}
\usepackage{algpseudocode}

\newtheorem{assumption}{Assumption}
\newtheorem{definition}{Definition}
\newtheorem{theorem}{Theorem}

\newcommand{\Agents}{\mathcal{I}}
\newcommand{\numAgents}{|\Agents|}
\newcommand{\A}{\mathcal{A}}
\newcommand{\B}{\mathcal{B}}
\newcommand{\Actions}{\mathcal{A}}
\newcommand{\numActions}{|\Actions|}

\newcommand{\Behaviors}{\mathcal{B}}
\newcommand{\numBehaviors}{|\Behaviors|}
\newcommand{\ones}{\mathbf{1}}
\newcommand{\zeroes}{\mathbf{0}}
\newcommand{\Ft}[1]{\mathcal{F}_{#1}}
\newcommand{\indep}{\rotatebox[origin=c]{90}{$\models$}}
\newcommand{\two}{\hspace{5px}}
\newcommand{\nn}{\nonumber}
\usepackage{xcolor}
\newcommand{\Reals}[1]{\mathbb{R}^{#1}}

\newcommand{\qlk}[1]{$\mathrm{QL}_{#1}$}
\newcommand{\QBR}[2]{\mathrm{QBR}(#1; #2)}

\newcommand{\Ex}[1]{\mathbb{E}\left(#1\right)}

\author[*]{Panos Toulis}
\author[**]{David C. Parkes}
\affil[*]{\small Econometrics and Statistics, University of Chicago, Booth School}
\affil[**]{\small Harvard University, School of Engineering and Applied Science}

\begin{document}
\title{Long-term causal effects via behavioral game theory}

\maketitle

\begin{abstract}
Planned experiments are the gold standard in reliably comparing the causal effect of switching from a baseline policy to a new policy.
One critical shortcoming of classical experimental methods, however, is that they typically do not take into account
the dynamic nature of response to policy changes.
For instance, in an experiment where we seek to understand the effects of a new ad pricing policy on auction revenue, agents may adapt their bidding in response to the experimental pricing changes.
Thus, causal effects of the new pricing policy after such adaptation period, the {\em long-term causal effects}, are
not captured by the classical methodology even though they clearly are more indicative of the value of the new policy.
 Here, we formalize a framework to define and estimate long-term causal effects of 
 policy changes in multiagent economies.
 Central to our approach is behavioral game theory, which we leverage 
 to formulate the ignorability assumptions that are necessary for causal inference. 
Under such assumptions we estimate long-term causal effects through a latent space approach, where a behavioral model of how agents act conditional on their latent behaviors is combined with a temporal model of how behaviors evolve over time.
\end{abstract}
\newpage

% Introduction.
\section{Introduction}
\label{section:intro}
A multiagent economy is comprised of agents interacting under specific economic rules.
A common problem of interest is to experimentally evaluate changes to such rules, also known as {\em treatments}, on an objective of interest.
For example, an online ad auction platform is a multiagent economy, 
where one problem is to estimate the effect of raising the reserve price on the platform's revenue.
Assessing causality of such effects is a challenging problem because there is a conceptual discrepancy between what needs to be estimated and what is available in the data, as illustrated in Figure~\ref{figure:tasks}.

What needs to be estimated is the {\em causal effect} of a policy change, which is 
defined as the difference between the objective value when the economy is treated, i.e., when {\em all} agents interact under the new rules, relative to when the same economy is in control, i.e., when {\em all} agents interact under the baseline rules. 
Such definition of causal effects is logically necessitated from the designer's task, which is to select either the treatment or the control policy based on their estimated revenues, and then apply such policy to all agents in the economy.
The {\em long-term causal effect} is the causal effect defined after the system has stabilized, and is more representative of the value of policy changes in dynamical systems. Thus, in Figure~\ref{figure:tasks} the long-term causal effect 
is the difference between the objective values at the top and bottom endpoints, marked as the ``targets of inference".

What is available in the experimental data, however, typically comes from designs such as the so-called A/B test, where we randomly assign {\em some} agents to the treated economy (new rules B) and the others to the control economy (baseline rules A), and then compare the outcomes. In Figure~\ref{figure:tasks} the experimental data are depicted as the solid time-series in the middle of the plot, marked as the ``observed data".
\begin{figure*}[t!]
\includegraphics[width=.9\textwidth]{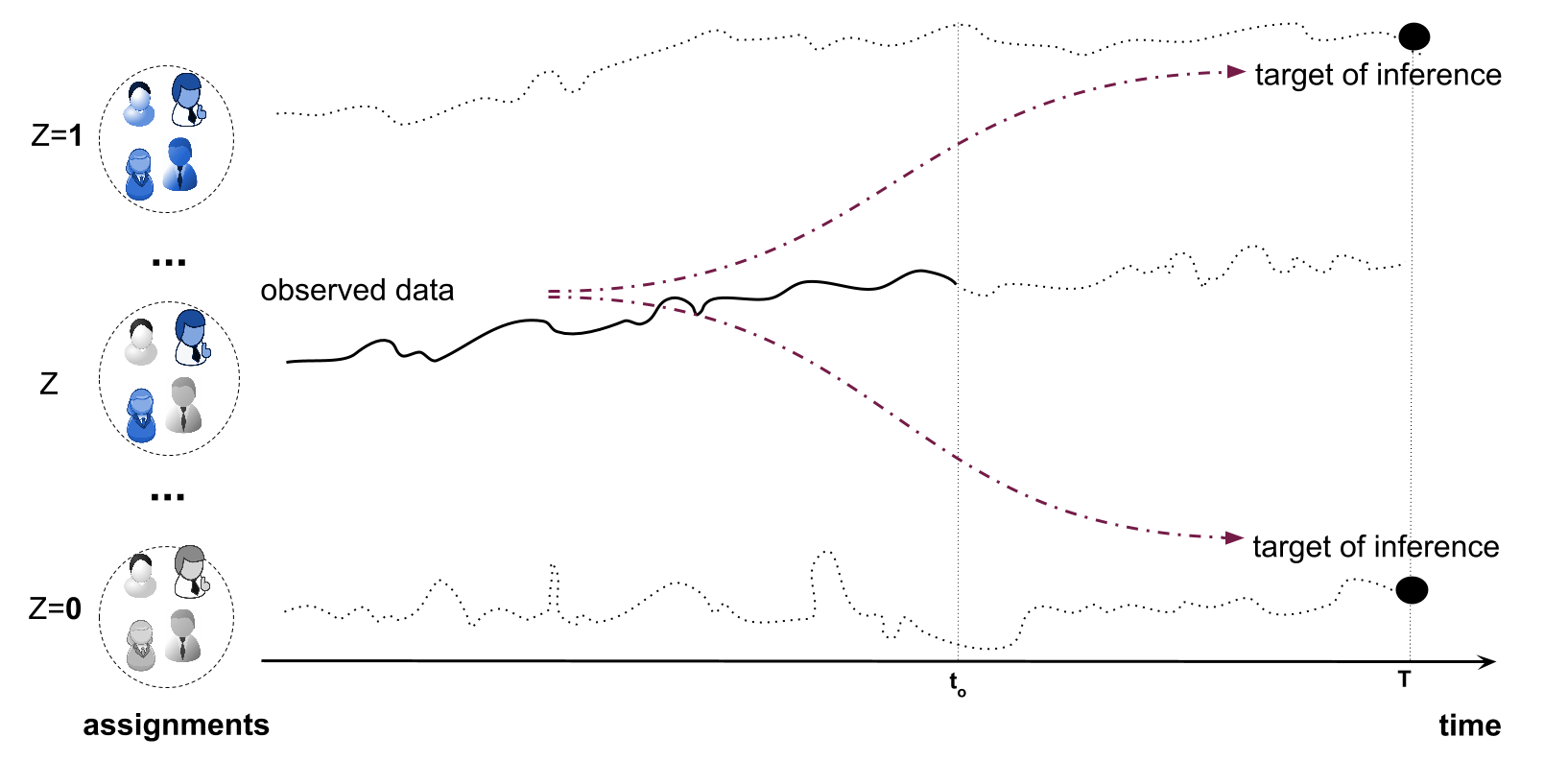}
\caption{\small The two inferential tasks for causal inference in multiagent economies.
First, infer agent actions across treatment assignments (y-axis), particularly, the assignment where all agents are in the treated economy (top assignment, $Z=1$), and
the assignment where all agents are in the control economy (bottom assignment, $Z=0$). 
Second, infer across time, from $t_0$ (last observation time) to long-term $T$. 
What we seek in order to evaluate the causal effect of the new treatment is the difference between the objectives (e.g., revenue) at the two inferential target endpoints.}
\label{figure:tasks}
\end{figure*}

Therefore the challenge in estimating long-term causal effects is that we generally need to perform two inferential tasks simultaneously, namely,
\begin{enumerate}[(i)]
    \item infer outcomes across possible experimental assignments (y-axis in Figure~\ref{figure:tasks}), and
    \item infer long-term outcomes from short-term experimental data (x-axis in Figure~\ref{figure:tasks}).
\end{enumerate}

The first task is commonly known as the ``fundamental problem of 
causal inference"~\citep{holland1986statistics, rubin2011causal} because it underscores the impossibility of observing in the same experiment the outcomes for {\em both} policy assignments that define the causal effect; i.e., that we cannot observe in the same experiment both the outcomes when all agents are treated and the outcomes when all agents are in control, the assignments of which are denoted by $Z=\textbf{1}$ and $Z=\textbf{0}$, respectively, in Figure~\ref{figure:tasks}.
% , 
% and implies that the ultimate goal of the experiment is to estimate the outcomes 
% when all experimental units, e.g., agents in the economy, are assigned to the new policy,
% and the outcomes when all units are assigned to the baseline policy.
% select either the treatment or control, which will later be adopted as the baseline economy for all agents.
% Therefore, the outcomes of interest correspond only to two  treatment assignments: outcomes measured when all agents  are assigned to treatment, and outcomes measured when all agents are assigned to control.
%
In fact the role of experimental design, as conceived by~\citet{fisher1935design}, is exactly
to quantify the uncertainty about such causal effects that cannot be observed due to the aforementioned fundamental problem, by using standard errors that can be observed in a carefully designed experiment. 

The second task, however, is unique to causal inference in dynamical systems, such as the multiagent economies that we study in this paper, and has received limited attention so far.
Here, we argue that it is crucial to study long-term causal effects, i.e., effects measured after the system has stabilized, because such effects are more representative of the value of policy changes. If our analysis focused only on the observed data part depicted in Figure~\ref{figure:tasks}, then policy evaluation would reflect transient effects that might differ substantially from the long-term effects.
For instance, raising the reserve price
in an auction might increase revenue in the short-term but as agents adapt their bids, or switch to another platform altogether, the long-term effect could be a net decrease in revenue~\citep{holland1991artificial}. 
%

%
% Central to the methodology we develop here is 
% to assume a {\em behavioral model} where agents adopt, at each time point, a latent {\em behavior} conditional on which the agent takes actions according to a distribution.
% %
% A {\em temporal model} is used for the evolution of agent behaviors, and the two models are fitted using the short-term experimental data. 
% The fitted model parameters are then used to predict the long-term agent actions, and
% thus estimate the long-term causal effect of interest, under certain ignorability assumptions that we state explicitly.
%
% However, our approach does not rely on a {\em particular} behavioral or temporal model. For example, in the application of Section \ref{section:application}, we use the quantal $k$-level (QLk), by Stahl and Wilson~\cite{stahl1994experimental}, as the behavioral model, and a lag-one vector autoregressive model, denoted by VAR(1), as the temporal model. More sophisticated models can certainly be used, however, these would still need to perform the inferential tasks of Figure \ref{figure:tasks}.

%% Related work
\subsection{Related work and our contributions}
\label{section:related}
There have been several important projects related to causal inference in multiagent economies.
For instance, \citet{ostrovsky2011} evaluated
the effects of an increase in the reserve price of Yahoo! ad auctions on revenue.
Auctions were randomly assigned to an
increased reserve price treatment, 
and the effect was estimated using  difference-in-differences (DID), which is a popular econometric method~\citep{card1993minimum,
  donald2007inference, ostrovsky2011}.
The DID method compares the difference in outcomes before and after the intervention for both the treated and control units ---the ad auctions in this experiment--- and then compares the two differences.
In relation to Figure~\ref{figure:tasks}, DID extrapolates across assignments (y-axis) and across time (x-axis)
by making a strong additivity assumption~\citep[Section 5.2]{abadie2005semiparametric, angrist2008mostly}, specifically, 
by assuming that the dependence of revenue on reserve price and time is additive.

In a structural approach, \citet{athey2008comparing} studied the effects of auction format (ascending versus
sealed bid) on competition for timber tracts. 
Their approach was to estimate agent valuations from observed data (agent bids) in one auction format and then  impute counterfactual bid distributions in the other auction format, under the assumption of equilibrium play in the observed data.
In relation to Figure~\ref{figure:tasks}, their approach extrapolates across assignments by assuming that agent individual valuations for tracts are independent of the treatment assignment,
and extrapolates across time by assuming that the observed agent bids are already in equilibrium. Similar approaches are followed in econometrics for estimation 
of general equilibrium effects~\citep{heckman1998general, heckman05}.

In a causal graph approach~\citep{pearl2000causality}, \citet{bottou2012} studied effects of changes in the algorithm that scores Bing ads on the ad platform's revenue.
Their approach was to create a directed acyclic graph (DAG) among related variables, such as queries, bids, and prices.  Through a ``Causal Markov'' assumption they could predict counterfactuals for  revenue, using only data from the control economy (observational study).
In relation to Figure~\ref{figure:tasks}, their approach is non-experimental and extrapolates across assignments and across time by assuming a directed acyclic graph (DAG) as the correct data model, which is also assumed to be stable with respect to treatment assignment, and by estimating counterfactuals through the fitted model.
% that the underlying DAG is the correct data model, and does not depend on the treatment assignment.
% This assumption of a well-defined DAG is crucial in transferring the estimates from the observed data to counterfactuals.
%
%

Our work is different from prior work
because it takes into account the short-term aspect of experimental data to evaluate long-term causal effects, which is the key conceptual and practical challenge that arises in empirical applications.
In contrast, classical econometric methods, such as DID, assume strong linear trends from short-term to long-term, whereas 
structural approaches typically assume that the experimental data are already long-term as they  are observed in equilibrium. We refer the reader to Sections 2 and 3 of the supplement for more detailed comparisons.

In summary, our key contribution is 
that we develop a formal framework that (i) articulates the distinction between short-term and long-term causal effects, (ii) leverages behavioral game-theoretic models for causal analysis of multiagent economies, and (iiii) explicates theory that enables valid inference of long-term causal effects.

%%
%%  Model and definitions.
\section{Definitions}
\label{section:model}
Consider a set of agents $\Agents$ and a set
of actions $\Actions$, indexed by $i$ and $a$, respectively. The experiment designer wants to run an experiment to evaluate a new policy against the baseline policy relative to an objective.
In the experiment each agent is assigned to one policy, and the experimenter observes how agents act over time.
Formally, 
let $Z=(Z_i)$ be the $\numAgents \times 1$ assignment vector
where $Z_i=1$ denotes that agent $i$ is assigned to the new policy, and $Z_i=0$ denotes that $i$ is assigned to the baseline policy; as a shorthand, $Z=\ones$ denotes that all agents are assigned to the new policy, and $Z=\zeroes$ denotes that all agents are assigned to the baseline policy,  
where $\ones$, $\zeroes$ generally denote an appropriately-sized vector of ones and zeroes, respectively.
In the simplest case, the experiment is an A/B test, where $Z$ is uniformly random on $\{0,1\}^{\numAgents}$ subject to $\sum_i Z_i=|\Agents|/2$.

After the initial assignment $Z$ agents play actions at discrete time points from $t=0$ to $t=t_0$.
Let $A_i(t; Z) \in \Actions$ be the random variable that denotes the action of agent $i$ at time $t$ under assignment $Z$.
The {\em population action} 
$\alpha_j(t; Z) \in \Delta^{\numActions}$, where $\Delta^p$ denotes the $p$-dimensional simplex, is the frequency of actions at time $t$ under assignment $Z$ of agents that were assigned to game $j$; for example, assuming two actions $\Actions=\{a_1, a_2\}$, then $\alpha_{1}(0; Z) = [0.2, 0.8]$ denotes that, under assignment $Z$, $20\%$ of agents assigned to the new policy play action $a_1$ at $t=0$, while the rest play $a_2$.
%
% i.e., 
%if $\Agents_j = \{\i \in \Agents : Z_i=j\}$ is the 
%set of agents assigned to game $\j$, then the
%$\a$th element of $\aggAction$ is equal to 
%$\sum_{\i \in \Agents_j} \mathbb{I}\{\poAction=a\} / |\Agents_j|$. As before, 
%$\aggAgentAction{\j}{\ast}{Z}$
%denotes the $\numActions \times t_0$ matrix where 
%$\aggAgentAction{\j}{k}{Z}$ is the $(k+1)$th column.
%
%\subsection{Long-term causal effect}
We assume that the objective value for the experimenter depends on the population action, in a similar way that, say, auction revenue depends on agents' aggregate bidding. The objective value in policy $j$ at time $t$ under assignment $Z$ is denoted by $R(\alpha_j(t; Z))$, 
where $R : \Delta^{\numActions} \to \mathbb{R}$.
For instance, suppose in the previous example that $a_1$ and $a_2$ produce revenue $\$10$ and $-\$2$, respectively, each time they are played, then $R$ is linear and $R([.2, .8]) = 0.2\cdot\$10 - 0.8\cdot\$2 = \$0.4$.
% It can immediately be seen that the objective depends on agent actions, in a similar way that, say, revenue depends on agent bids. 
%
\begin{definition}
\label{definition:ce}
The average causal effect on objective $R$ at time $t$ of the new policy relative to the baseline is denoted by $\mathrm{CE}(t)$ and is defined as
\begin{align}
\label{eq:ce}
\mathrm{CE}(t) & = \Ex{R(\alpha_1(t; \ones)) - R(\alpha_0(t; \zeroes))}.
\end{align}
\end{definition}

Suppose that $(t_0, T]$ is the time interval required for the economy to adapt to the experimental conditions. The exact definition of $T$ is important but we defer this discussion for Section~\ref{section:discussion}.
The designer concludes that the new policy is better than the baseline if $\mathrm{CE}(T) > 0$.
Thus, $\mathrm{CE}(T)$ is the 
{\em long-term average causal effect}
and  is a function 
of two objective values, $R(\alpha_1(T; \ones))$ and $R(\alpha_0(T; \zeroes))$, which correspond to the two inferential target endpoints 
in Figure~\ref{figure:tasks}.
Neither value is observed in the experiment because agents are randomly split between policies, and their actions are observed only for the short-term period $[0, t_0]$. Thus we need to (i) extrapolate across assignments by pivoting from the observed assignment to the counterfactuals $Z=\ones$ and $Z=\zeroes$;
(ii) extrapolate across time from the short-term data $[0, t_0]$ to the long-term $t=T$. 
We perform these two extrapolations based on a latent space approach, which is described next.

\subsection{Behavioral and temporal models}
We assume a latent behavioral model of how agents select actions, inspired by models from behavioral game theory.
The behavioral model is used to predict agent actions conditional on agent behaviors, and is combined with a temporal model to predict behaviors in the long-term.
The two models are ultimately used to 
estimate agent actions in the long-term, 
and thus estimate long-term causal effects.
As the choice of the latent space is not unique, in Section~\ref{section:discussion} we discuss why we chose to use behavioral models from game theory.

Let
$B_i(t; Z)$ denote the behavior that agent $i$ adopts at time $t$ under experimental assignment $Z$.
The following assumption puts a constraints on the space of possible behaviors that agents can adopt, 
which will simplify the subsequent analysis.
\begin{assumption}[Finite set of possible behaviors]
\label{assumption:finite_behaviors}
There is a fixed and finite set of behaviors $\Behaviors$ 
such that for every time $t$, assignment $Z$ and agent $i$, it holds that $B_i(t; Z) \in \Behaviors$; i.e., every agent can only adopt a behavior from $\Behaviors$.
\end{assumption}

The set of possible behaviors $\Behaviors$ essentially defines a $\numBehaviors \times \numActions$ collection of probabilities that is sufficient to compute the likelihood of actions played conditional on adopted behavior---we refer to such collection as the behavioral model.

\begin{definition}[Behavioral model]
\label{definition:behavioral}
The behavioral model for policy $j$ defined by set $\Behaviors$ of behaviors is the collection of probabilities
\begin{align}
P(A_i(t; Z)=a|B_i(t; Z)=b, G_j),\nonumber
\end{align}
for every action $a\in\Actions$ and every behavior $b\in\Behaviors$,
where $G_j$ denotes the characteristics of policy $j$.
\end{definition}

As an example, a non-sophisticated behavior $b_0$ could imply that
$P(A_i(t; Z)=a|b_0, G_j) = 1/\numActions$, i.e., that the agent adopting $b_0$ simply plays actions at random.
Conditioning on policy $j$ in Definition~\ref{definition:behavioral} allows an agent to choose its actions based on expected payoffs, which
depend on the policy characteristics. For instance, in the application
of Section~\ref{section:application} we consider a behavioral model
where an agent picks
actions in a two-person game according to expected payoffs calculated from the
game-specific payoff matrix---in that case $G_j$ is simply the payoff matrix of game $j$.

The {\em population behavior} 
$\beta_j(t; Z) \in \Delta^{\numBehaviors}$ denotes the frequency at time $t$ under assignment $Z$ of the adopted  behaviors of agents assigned to policy $j$.
Let $\Ft{t}$ denote the entire history of population behaviors in the experiment up to time $t$. 
A temporal model of behaviors is defined as follows.

\begin{definition}[Temporal model]
For an experimental assignment $Z$ a temporal model for policy $j$ is a collection of parameters $\phi_j(Z), \psi_j(Z)$, and densities $(\pi, f)$, such that for all $t$,
\begin{align}
\label{eq:temporal}
\beta_j(0; Z) & \sim \pi(\cdot; \phi_j(Z)),
\nn \\
\beta_j(t; Z) |  \two \Ft{t-1}, G_j
 & \sim f(\cdot | \psi_j(Z), \Ft{t-1}). 
\end{align}
\end{definition}

A temporal model defines the distribution of population behavior as a time-series with a Markovian structure subject to $\pi$ and $f$ being stable with respect to $Z$. In other words, regardless of how agents are assigned to games, the population behavior in the game will evolve according to a fixed model described by $f$ and $\pi$. The model parameters $\phi, \psi$ may still depend on the treatment assignment $Z$.
%

%% Main assumption.
\section{Estimation of long-term causal effects}
\label{section:theory}

%The behavioral and temporal models define, essentially, a hidden Markov model over time of agent actions. 
% Conceptually, to predict the long-term actions at time $T$ in game $j$ under assignment $Z$, we (i) estimate the latent behavior at $\aggAgentBehavior{j}{t_0}{Z}$ using the forward algorithm \citep{bishop2006pattern}, then (ii) estimate $\aggAgentBehavior{j}{T}{Z}$ from $\aggAgentBehavior{j}{t_0}{Z}$ and the temporal model, and finally (iii) estimate $\aggAgentAction{j}{T}{Z}$ conditional on $\aggAgentBehavior{j}{T}{Z}$.
%
% Thus, so far we have only established inference across time (i.e., along the x-axis in Figure \ref{figure:tasks}), but not inference across assignments (y-axis in Figure \ref{figure:tasks}).
%
Here we develop the assumptions that are necessary for inference of long-term causal effects.
%%   Main assumption
%% Need a name --- Ignorability ?
\begin{assumption}[Stability of initial behaviors]
\label{assumption:initial}
Let $\rho_Z = \sum_{i\in\Agents} Z_i / \numAgents$ be the proportion of agents assigned to the new policy under assignment $Z$.
Then, for every possible $Z$,
\begin{align}
\label{eq:invariant}
\rho_Z \beta_1(0; Z) + (1-\rho_Z) \beta_0(0; Z) = \beta^{(0)},
\end{align}
where $\beta^{(0)}$ is a fixed population behavior invariant to $Z$.
\end{assumption}

\begin{assumption}[Behavioral ignorability]
\label{assumption:ignorability}
The assignment is independent of population behavior at time $t$,
conditional on policy and behavioral history up to $t$; i.e., 
for every $t>0$ and policy $j$,
\begin{align}
Z \two \indep \two \beta_j(t; Z) \two | \two \Ft{t-1}, G_j. \nn
\end{align}
\end{assumption}

{\em Remarks.} 
% Since $\beta_0(t; \ones) = \zeroes$, by definition, Assumption~\ref{assumption:initial} implies that 
% $\beta_1(0; \ones) = \beta^{(0)}$; similarly, $\beta_0(t; \zeroes)  = \beta^{(0)}$. 
Assumption~\ref{assumption:initial} implies that the agents do not anticipate the assignment $Z$ as they ``have made up their minds" to adopt a population behavior $\beta^{(0)}$ before the experiment.
It follows that the population behavior $\beta_1(t; Z)$ marginally corresponds to $\rho_Z \numAgents$ draws from $\numBehaviors$ bins of total size $\numAgents \beta^{(0)}$. The bin selection probabilities at every draw depend on the experimental design; for instance, in an A/B experiment where $\rho_Z=0.5$ the population behavior at $t=0$ can be sampled uniformly such that $\beta_1(0; Z) + \beta_0(0; Z) = 2 \beta^{(0)}$.
Quantities such as that in Eq.~\eqref{eq:invariant}
are crucial in causal inference because they can be used as a pivot for extrapolation across assignments.
%that are stable under the treatment assignment $Z$ 
%

Assumption \ref{assumption:ignorability} states that the treatment assignment 
does not add information about the population behavior at time $t$, if 
we already know the full behavioral history of up to $t$, and the policy which agents are assigned to; hence, the treatment assignment is conditionally {\em ignorable}.
This ignorability assumption  precludes, for instance, an agent  adopting a different 
behavior depending on whether it was assigned with friends or foes in the experiment.

%For an intuition why Assumptions \ref{assumption:initial} and \ref{assumption:ignorability} enable inference of the long-term causal effect note that Assumption~\ref{assumption:initial} implies that the prior parameter $\phi_Z^j$ does not depend on $Z$, since agents start with a fixed aggregate behavior $\bo_j$, and they are randomly assigned to games. 
% Furthermore, Assumption~\ref{assumption:ignorability} implies that $\psi_Z^j$ also does not depend on $Z$ in Eq.~\eqref{eq:temporal} because the treatment assignment is independent of the aggregate behavior at $t$, conditional on the history and the game.
%
Algorithm~\ref{estimation_algo} is the main methodological contribution of this paper.
It is a Bayesian procedure as it puts priors on parameters $\phi, \psi$ of the temporal model, and then marginalizes these parameters out.

\newcommand{\algoComment}[1]{
{\em \hspace{1px} \#\textcolor{gray}{#1}}}
\begin{algorithm}[h!]
\setstretch{2}
\caption{Estimation of long-term causal effects.\newline
{\bf Input:} $Z, T, \A, \B, G_1, G_0, \mathcal{D}_1 = \{a_1(t; Z) : t =0, \ldots, t_0\}, 
\mathcal{D}_0 = \{a_0(t; Z) :  t =0, \ldots, t_0\}$ \newline
{\bf Output:} Estimate of long-term causal effect $\mathrm{CE}(T)$ in Eq.~\eqref{eq:ce}.
}
\label{estimation_algo}
%-----------------------------------------------------------------------------------------------------------------------
\begin{algorithmic}[1]
\State By Assumption~\ref{assumption:ignorability}, define 
$\phi_j \equiv \phi_j(Z)$, $\psi_j \equiv \psi_j(Z)$.
\State Set $\mu_1 \leftarrow \zeroes$ and
$\mu_0 \leftarrow \zeroes$, both of size $\numActions$;
set $\nu_0=\nu_1=0$.
\For{$iter=1, 2, \ldots$}
\State For $j=0, 1$, sample $\phi_j, \psi_j$ from prior, 
and sample $\beta_j(0; Z)$ conditional on $\phi_j$.
\State Calculate $\beta^{(0)} = \rho_Z \beta_1(0; Z) + (1-\rho_Z) \beta_0(0; Z)$.
\For{$j= 0, 1$}
\State Set $\beta_j(0; j\ones) = \beta^{(0)}$.
% algoComment{pivot to $Z=\ones$ and $Z=\zeroes$.}
\State Sample $B_j = \{\beta_j(t; j\ones) : t=0,\ldots,T\}$ given $\psi_j$ and $\beta_j(0, j\ones)$.
\algoComment{temporal model}
\State Sample $\alpha_j(T; j\ones)$ conditional on $\beta_j(T; j\ones)$. \algoComment{behavioral model}
% \State Calculate $w = P\left(\mathcal{D}_j | G_j, \{\beta_j(t; j\ones) :  t\in[0, t_0]\}\right)$.
\State  Set $\mu_j \leftarrow \mu_j + P\left(\mathcal{D}_j |  B_j, G_j\right) \cdot R(\alpha_j(T; j\ones))$. 
\State  Set $\nu_j \leftarrow \nu_j + P\left(\mathcal{D}_j | B_j, G_j\right)$.
\EndFor
\EndFor
\State Return estimate
$
\widehat{\mathrm{CE}}(T) = \mu_1 / \nu_1 - \mu_0 / \nu_0.
$
\end{algorithmic}
\end{algorithm}

%% FIGURE Of estimation
% (SAVE SPACE)
%A weighted sum of these sampled actions is maintained in the vector $\mu_j$, which  is shown in Theorem~\ref{theorem:main} to be an unbiased estimate  of $\alpha_j(T; j\ones)$, i.e, the long-term population action when  all agents are assigned to policy $j$.
% Finally, these action estimates can be used to estimate 
% Thus, estimation of the long-term causal effect can be done by (i) estimating the initial  aggregate behavior at $t=0$ under the realized assignment in the experiment, (ii) estimating  the aggregate behavior  t $t=0$ under the counterfactual assignment where all agents are assigned to a single game, and (iii) estimating the aggregate behavior and action at $t=T$ through the temporal and behavioral models. This estimation process is illustrated by the path (A)-(B)-(C)-(D)-(E) on the figure of Section 5 in the supplement. A concrete implementation of Algorithm~\ref{estimation_algo} is described in Section~\ref{section:application} of this paper.

% It is important to note that Algorithm 1 can either be viewed as producing point estimates, or as producing posterior distributions conditional on the observed data. For example, in Step 2 we can either estimate the initial aggregate behavior using the maximum-likelihood estimate given the observed data, or take its posterior distribution conditional on  these actions, assuming appropriate priors. The same holds for all subsequent steps. 
%

\begin{theorem}[Estimation of long-term causal effects]
\label{theorem:main}
  Suppose that behaviors evolve according to a known temporal model, and actions are 
  distributed conditionally on behaviors according to a known behavioral model.
  Suppose that Assumptions~\ref{assumption:finite_behaviors}, \ref{assumption:initial} and \ref{assumption:ignorability} hold for such models.
  Then, for every policy $j\in\{0, 1\}$ as the iterations of Algorithm~\ref{estimation_algo} increase,
  $
  \mu_j/\nu_j \to \Ex{R(\alpha_j(T; j\ones)) | \mathcal{D}_j}.
  $
  The output $\widehat{\mathrm{CE}}(T)$ of Algorithm~\ref{estimation_algo} asymptotically estimates the long-term causal effect, i.e.,
  $$
  \mathbb{E}(\widehat{\mathrm{CE}}(T))
   = \Ex{R(\alpha_1(T; \ones)) -R(\alpha_0(T; \zeroes))} \equiv \mathrm{CE}(T).
  $$
  %i.e., Algorithm~\ref{estimation_algo} unbiasedly estimates the long-term causal effect of interest in Eq.~\eqref{eq:ce}.
\end{theorem}

{\em Remarks.}
Theorem~\ref{theorem:main} shows that $\widehat{\mathrm{CE}}(T)$ consistently estimates the long-term causal effect in Eq.~\eqref{eq:ce}.
We note that it is also possible to derive the variance of this estimator with respect to the randomization distribution of assignment $Z$. To do so we first create a set of assignments $\mathcal{Z}$ by repeatedly sampling $Z$ according to the experimental design. Then we adapt  Algorithm~\ref{estimation_algo} so that (i) Step 4 is removed;
(ii) in Step 5, $\beta^{(0)}$ is sampled from its posterior distribution conditional on observed data, which can be obtained from the original Algorithm~\ref{estimation_algo}.
The empirical variance of the outputs over $\mathcal{Z}$ from the adapted algorithm estimates the variance of the output $\widehat{\mathrm{CE}}(T)$ of the original algorithm. We leave the full characterization of this variance estimation procedure for future work.

%% (SAVE SPACE)
% Theorem 1 shows the properties of Algorithm 1 for estimation of long-term
% causal effects---the proof is given in the supplement.
% The key property is that the quantity $\mu_j/\ones'\mu_j$ 
% unbiasedly estimates the long-term  population action in policy $j$ when all agents are assigned to policy $j$, conditional on the short-term  population actions observed up to $t_0$.
%The proof, given in Sections 1 and 2 of the supplement, relies on being able to identify the parameters $(\phi^j, \psi^j)$, for each game $j$.  This is possible  when there is enough available data (e.g., through multiple experiments), or when the observation period $[0, t_0]$ increases.
% The proof follows the Bayesian interpretation of Algorithm 1 and is given in the supplementary material.
%
As Theorem 1 relies on Assumptions~\ref{assumption:initial} and \ref{assumption:ignorability}, it is worth noting that the assumptions may be hard but not impossible to test in practice. For example, one idea to test Assumption~\ref{assumption:ignorability} is to use data from multiple experiments on a single game $j$. If fitting the temporal model~\eqref{eq:temporal} on such data yields parameter estimates $(\phi_j(Z), \psi_j(Z))$ that depend on experimental assignment $Z$, then Assumption~\ref{assumption:ignorability} would be unjustified. A similar test could be used for Assumption~\ref{assumption:initial} as well.

\subsection{Discussion}
\label{section:discussion}
Methodologically, our approach is aligned with the idea that for long-term causal effects we need a model for outcomes that leverages structural information pertaining to how outcomes are generated and how they evolve. In our application such structural information is the microeconomic information that dictates what agent behaviors are successful in a given policy and how these behaviors evolve over time.

In particular, Step 1 in the algorithm relies on Assumptions~\ref{assumption:initial} and \ref{assumption:ignorability} to infer that model parameters, $\phi_j, \psi_j$ are stable with respect to treatment assignment.
Step 5 of the algorithm is the key estimation pivot, which 
uses Assumption~\ref{assumption:initial} to extrapolate from the experimental assignment $Z$ to the counterfactual assignments 
$Z=\ones$ and $Z=\zeroes$, as required in our problem. 
Having pivoted to such counterfactual assignment, it is then possible to use the temporal model parameters $\psi_j$, which are unaffected by the pivot under Assumption~\ref{assumption:ignorability}, to 
sample population behaviors up to long-term $T$, and subsequently sample agent actions at $T$ (Steps 8 and 9). 

Thus, a lot of burden is placed on the behavioral game-theoretic model to predict agent actions, and the accuracy of such models is still not settled~\citep{hahn2015bayesian}. However, it does not seem necessary that such prediction is completely accurate, but rather that the behavioral models can pull relevant information from data that would otherwise be inaccessible without game theory, thereby improving over classical methods. A formal assessment of such improvement, e.g., using information theory, is open for future work. 
An empirical assessment can be supported by the extensive literature in behavioral game
theory~\citep{stahl1994experimental, mckelvey95}, which has been successful in predicting human actions in real-world experiments~\citep{wright2010beyond}.
%Such combination of behavioral game theory with causal inference is novel and crucial, since better behavioral models will generally support better causal inference.

% A question that now opens up is what makes our approach causal and not standard machine learning, since it relies so heavily on statistical behavioral models? 
% %
% The answer to this question is twofold.
% First, Algorithm~\ref{estimation_algo} estimates the quantity in Eq.~\eqref{eq:ce}, which is defined in terms of counterfactuals and therefore \textit{is} a causal effect, by definition.
% Second, Algorithm~\ref{estimation_algo} takes into account the randomization from the experimental design, which is crucial in statistical causal inference.
% Specifically, Step 5 of Algorithm~\ref{estimation_algo} encodes how the initial behaviors are generated  based on the experimental design and on Assumption~\ref{assumption:initial}, which posits that agents start with a fixed but unknown population behavior.
% In our setting the design is a simple A/B test but more elaborate designs could be handled by appropriately adapting Assumption~\ref{assumption:initial} and Step 5. For instance, if we had a blocked experiment we could extend Assumption~\ref{assumption:initial} to have different initial behaviors $\beta^{(0)}$ per block.

Another limitation of our approach is Assumption~\ref{assumption:finite_behaviors}, which posits that there is a finite set of predefined behaviors. A nonparametric approach where behaviors are estimated on-the-fly might do better.
In addition, the long-term horizon, $T$, also needs to be defined {\em a priori}.
We should be careful how $T$ interferes with the temporal model since such a model implies a time $T'$ at which population behavior reaches stationarity. Thus if $T' \le T$ we implicitly assume that the long-term causal effect of interest pertains to a stationary regime (e.g., Nash equilibrium), 
but if $T' > T$ we assume that the effect pertains to a transient regime, and therefore the policy evaluation might be misguided.

\section{Application on data from a behavioral experiment}
\label{section:application}
In this section, we apply our methodology
to experimental data from~\citet{rapoport1992mixed}, as reported by~\citet{mckelvey95}.
The experiment consisted of a series of zero-sum two-agent games, and aimed at examining the hypothesis that human players play according to minimax solutions of the game, the so-called 
 minimax hypothesis initially suggested by~\citet{von1944theory}.
Here we repurpose the data in a slightly artificial way, including how we construct the designer's objective. This enables a suitable demonstration of our approach.

Each game in the experiment was a simultaneous-move game with five discrete 
actions for the row player and five actions for the column player. The structure of the payoff matrix, given in the supplement in Table 1, 
is parametrized by two values, namely $W$ and $L$;
the experiment used two
different versions of payoff matrices, 
corresponding to payments by the row agent to the column agent  
when the row agent {\em won} ($W$), or {\em lost} ($L$):
modulo a scaling factor~\citet{rapoport1992mixed} used $(W, L) =(\$10, -\$6)$ for game 
0 and $(W, L) =(\$15, -\$1)$ for game 1. 
% For example, if the row agent picks action $\a_1$ and the column agent plays $\a_3'$ in game 0,  then the row agent has to pay $\$6$ to the column agent.

Forty agents, $\Agents=\{1, 2, \ldots, 40\}$, were randomized to
one game design (20 agents per game), and each agent
played once as row and once as column, matched against two
different agents. 
Every match-up between a pair of
agents lasted for two periods of 60 rounds, 
with each round consisting of a selection of an action from
each agent and a payment. Thus, each agent played
for four periods and 240 rounds in total.
If $Z$ is the entire assignment vector of length 40, $Z_i=1$ means that agent $i$ was assigned to 
game 1 with payoff matrix $(W, L) =(\$15, -\$1)$ and $Z_i=0$ means that $i$ was assigned to 
game 0 with payoff matrix $(W, L) =(\$10, -\$6)$.

In adapting the data, we take advantage of the randomization in the experiment, 
and ask a question in regard to long-term causal effects. 
In particular, assuming that agents pay a fee for each action taken, 
which accounts for the revenue of the game, we ask the following question: 

\begin{quote}
What is the long-term
causal effect on  revenue if we switch from payoffs $(W, L) =(\$10, -\$6)$
of game 0 to payoffs $(W, L) =(\$15, -\$1)$ of game 1?".
\end{quote}
The games induced by the two aforementioned payoff matrices represent the two different policies we wish to compare.
To evaluate our method, we consider the last period as long-term, and hold out data from this period.
We define the causal estimand in Eq.~\eqref{eq:ce} as
\begin{align}
\label{eq:estimand_rapoport}
\mathrm{CE} = c^\intercal (\alpha_1(T; \ones) - 
\alpha_0(T; \zeroes)),
\end{align}
where $T=3$ and $c$ is a vector of coefficients.
The interpretation  is that, given an element $c_a$ of $c$, the agent playing action $a$ 
is assumed to pay a constant fee $c_a$.
To check the robustness of our method we test Algorithm 1 over multiple values of $c$.

%% Proposed method.
\subsection{Implementation of Algorithm~\ref{estimation_algo} and results}
\label{section:concrete}
Here we demonstrate how Algorithm~\ref{estimation_algo} can be applied to estimate the long-term causal effect in Eq.~\eqref{eq:estimand_rapoport} on the Rapoport \& Boebel dataset. To this end we clarify Algorithm~\ref{estimation_algo} step by step, and give more details in the supplement.

\textbf{Step 1: Model parameters.} 
For simplicity we assume that the models in the two games share common parameters, and thus 
$(\phi_1, \psi_1, \lambda_1) = (\phi_0, \psi_0, \lambda_0)\equiv (\phi, \psi, \lambda)$, where $\lambda$ are the parameters of the behavioral model to be described in Step 8. Having common parameters also acts as regularization and thus helps estimation.

\textbf{Step 4: Sampling parameters and initial behaviors} As explained later we assume that there are 3 different behaviors and thus 
$\phi, \psi, \lambda$ are vectors with 3 components. Let $x\sim U(m, M)$ denote that every component 
of $x$ is uniform on $(m, M)$, independently. We choose diffuse priors for our parameters, 
specifically, $\phi\sim \mathrm{U}(0, 10)$, $\psi\sim \mathrm{U}(-5, 5)$, and $\lambda\sim \mathrm{U}(-10, 10)$.
Given $\phi$ we sample the initial behaviors as Dirichlet, i.e., $\beta_1(0; Z)\sim \mathrm{Dir}(\phi)$ and $\beta_0(0; Z)\sim\mathrm{Dir}(\phi)$, independently.

\textbf{Steps 5 \& 7: Pivot to counterfactuals.} 
Since we have a completely randomized experiment (A/B test) it holds that $\rho_Z=0.5$ and 
therefore $\beta^{(0)} = 0.5 (\beta_1(0; Z) + \beta_0(0; Z))$.
Now we can pivot to the counterfactual population behaviors under $Z=\ones$ and $Z=\zeroes$ by setting 
$\beta_1(0; \ones) = \beta_0(0; \zeroes) = \beta^{(0)}$.

\textbf{Step 8: Sample counterfactual behavioral history.}
As the temporal model, we adopt
the {\em lag-one vector autoregressive model}, also known as $\mathrm{VAR}(1)$.
We transform\footnote{$y = \mathrm{logit}(x)$ is defined as the function $\Delta^m \to \Reals{m-1}$, $y[i]=\mathrm{log}(x[i+1]/x[1])$, where $x[1] \neq 0$ wlog.} 
the population behavior into a new variable $w_t = \mathrm{logit}(\beta_1(t; \ones))\in\Reals{2}$ (also do so for $\beta_0(t; \zeroes)$). Such transformation with a unique inverse is necessary because population behaviors are constrained on the simplex, and thus form so-called compositional data~\citep{aitchison1986statistical, grunwald1993time}. The VAR(1) model implies that
\begin{align}
w_t = \psi[1] \ones + \psi[2] w_{t-1} + \psi[3] \epsilon_t,\nonumber
\end{align}
where $\psi[k]$ is the $k$th component of $\psi$ and $\epsilon_t \sim \mathcal{N}(0, I)$ is i.i.d. standard bivariate normal. Eq.~\eqref{eq:var} is used to sample 
the behavioral history, $B_j$, in Step 8 of Algorithm~\ref{estimation_algo}.
%We note that the VAR model predicts the temporal evolution of counterfactual population behaviors when $Z=\ones$ or $Z=\zeroes$, and not for any $Z$.

\textbf{Step 9: Behavioral model.}
For the behavioral model, we adopt the \emph{quantal $p$-response} (\qlk{p}) model~\citep{stahl1994experimental}, which has been successful in predicting human actions in real-world experiments~\citep{wright2010beyond}. 
We choose $p=3$ behaviors, namely $\B=\{b_0, b_1, b_2\}$ of increased sophistication parametrized by $\lambda = (\lambda[1], \lambda[2], \lambda[3])\in\Reals{3}$.
%
% \qlk{3} defines the following behavioral model.
%
Let $G_j$ denote the $5\times 5$ payoff matrix of game $j$ and let the term {\em strategy} denote a distribution over all actions. 
An agent with behavior $b_0$ plays the uniform strategy, $$
P(A_i(t; Z)=a| B_i(t; Z)=b_0, G_j) = 1/5.
$$
An agent of level-1 (row player) assumes to be playing only against level-0 agents and thus
expects per-action profit $u_1 = (1/5) G_j \ones$ (for column player we use the transpose of $G_j$). 
The level-1 agent will then play a strategy proportional to $e^{\lambda[1] u_1}$, where $e^x$ for vector $x$ denotes the element-wise exponentiation, $e^x = (e^{x[k]})$.
% % Let $u_a$ be the expected utility if an agent plays action $a$ against a level-0 agent; i.e., 
% if the agent is a row player we can write
% $u_a = (1/3) G_{ja}^\intercal \ones$, where $G_{ja}$ is the $a$th row of the payoff matrix $G_j$ (the column player case is symmetric). Then, an agent adopting behavior $b_1$ will play action $a$ with probability
% \begin{align}
% \label{eq:qlk1}
% P(A_i(t; Z)=a|B_i(t; Z)=b_1, G_j) \propto e^{u_a \lambda[1]}.
% \end{align}
The precision parameter $\lambda[1]$ determines 
how much an agent insists on maximizing expected utility; 
for example, if $\lambda[1]=\infty$, the agent plays the action 
with maximum expected payoff (best response); if $\lambda[1]=0$, the agent acts as a level-0 agent.
An agent of level-2 (row player) assumes to be playing only against level-1 agents with precision $\lambda[2]$ and therefore expects to face strategy proportional to $e^{\lambda[2] u_1}$.~Thus its expected per-action profit is $u_2 \propto G_j e^{\lambda[2] u_1}$, and plays strategy $\propto e^{\lambda[3] u_2}$.

Given $G_j$ and $\lambda$ we calculate a $5\times 3$ matrix $Q_j$ where the $k$th column is the strategy played by an agent with behavior $b_{k-1}$. The expected population action is therefore $\bar{\alpha}_j(t; Z) = Q_j \beta_j(t; Z)$.  
The population action $\alpha_j(t; Z)$ is distributed as a normalized multinomial random variable with expectation $\bar{\alpha}_j(t; Z)$, and so 
$P(\alpha_j(t; \ones) | \beta_j(t; \ones), G_j) = \mathrm{Multi}(|\Agents| \cdot \alpha_j(t; \ones); \bar{\alpha}_j(t; \ones))$, where $\mathrm{Multi}(n; p)$ is the multinomial density of observations $n=(n_1, \ldots, n_K)$ with probabilities $p=(p_1, \ldots, p_K)$. Hence, the full likelihood for observed actions in game $j$ in Steps 10 and 11 of Algorithm~\ref{estimation_algo} is given by the product
$$
P(\mathcal{D}_j | B_j, G_j) = \prod_{t=0}^{T-1} \mathrm{Multi}(|\Agents| \cdot \alpha_j(t; j\ones); \bar{\alpha}_j(t; j\ones)).
$$

\begin{figure}[t!]
\centering
\hspace{-10px}
\includegraphics[width=.65\textwidth]{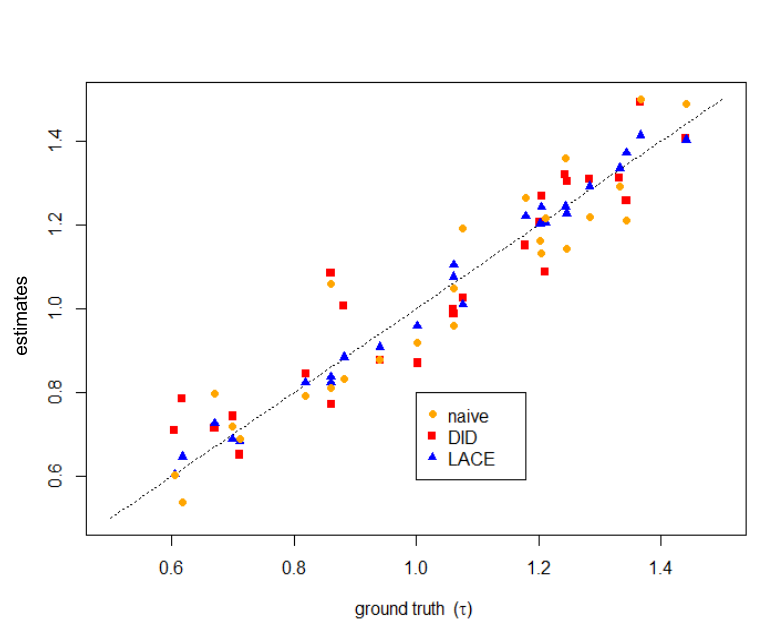}
\caption{Estimates of long-term effects from different methods
corresponding to 25 random objective coefficients $c$ in Eq.~\eqref{eq:estimand_rapoport}. 
For estimates of our method we ran  Algorithm~\ref{estimation_algo} for 100 iterations.}
\label{figure:results}
\end{figure}
Running Algorithm~\ref{estimation_algo} on the Rapoport and Boebel dataset  yields the estimates shown in Figure \ref{figure:results}, for 25 different fee vectors $c$, where each component $c_a$ is sampled uniformly at random from $(0, 1)$.
We also test
difference-in-differences (DID), which estimates the causal effect through 
$$
\hat{\tau}^{did} = 
[R(\alpha_{1}(2; Z)) - R(\alpha_1(0; Z))] - [R(\alpha_{0}(2; Z)) - R(\alpha_0(0; Z))],
$$ and a naive method (``naive" in the plot), which ignores the dynamical aspect and estimates the long-term causal effect as $\hat{\tau}^{nai} = [R(\alpha_{1}(2; Z)) - R(\alpha_0(2; Z))]$. 
Our estimates (``LACE" in the plot) are closer to the truth ($\mathrm{mse}=0.045$) than the estimates from the naive method ($\mathrm{mse}=0.185$) and from DID ($\mathrm{mse}=0.361$). 
This illustrates that our method can pull game-theoretic information from the data for long-term causal inference, whereas the other methods cannot.
% As the first paper on long-term causal effects, we think the additional complexity would obscure our conceptual and methodological contributions.

%
%\CenterFloatBoxes
%\begin{floatrow}
%\ffigbox
% 
%\killfloatstyle
\section{Conclusion}
\label{section:conclusion}
One critical shortcoming of statistical methods of causal inference is that they typically do not assess the long-term effect of policy changes.
Here we combined causal inference and game theory to build a framework for estimation of such long-term effects in multiagent economies.
 Central to our approach is behavioral game theory, which provides a natural latent space model of how agents act and how their actions evolve over time. Such models enable to predict how agents would act under various policy assignments and at various time points, which is key for valid causal inference. 
%
% The methodology relies on a set of assumptions 
 Working on data from an actual behavioral experiment set we showed how our framework can be applied to estimate the long-term effect of changing the payoff structure of a normal-form game.

Our framework could be extended in future work by incorporating learning (e.g., fictitious play, bandits, no-regret learning) to better model the dynamic response of multiagent systems to policy changes. 
Another interesting extension would be to use our framework for optimal design of experiments in such systems, 
which needs to account for heterogeneity in agent learning capabilities and for intrinsic dynamical  properties of the systems' responses to experimental treatments.
% We believe that progress in addressing these issues will lead to new and fruitful interactions of game theory with experimental design and causal inference.

\section*{Acknowledgements}
The authors wish to thank Leon Bottou, the organizers and participants of CODE{@}MIT'15, GAMES'16, the Workshop on Algorithmic Game Theory and Data Science (EC'15), and the anonymous NIPS reviewers for their valuable feedback.
Panos Toulis has been supported in part by the 2012 Google US/Canada Fellowship in Statistics.
David C. Parkes was supported in part by NSF grant CCF-1301976 and the SEAS TomKat fund.

\bibliography{refs}
\bibliographystyle{natbib}
\pagebreak

\appendix

\newcommand{\commentEq}[1]{\small\hspace{10pt} \text{[{\em #1} ]}}
\newcommand{\holdout}[1]{\textcolor{lightgray}{#1}}

\setcounter{theorem}{0}
\section{Proof of Theorem~\ref{theorem:main}}
\begin{theorem}[Estimation of long-term causal effects]
\label{theorem:main}
  Suppose that behaviors evolve according to a known temporal model, and actions are 
  distributed conditionally on behaviors according to a known behavioral model.
  Suppose that Assumptions~\ref{assumption:finite_behaviors}, \ref{assumption:initial} and \ref{assumption:ignorability} hold for such models.
  Then, for every policy $j\in\{0, 1\}$ as the iterations of Algorithm~\ref{estimation_algo} increase,
  $
  \mu_j/\nu_j \to \Ex{R(\alpha_j(T; j\ones)) | \mathcal{D}_j}.
  $
  The output $\widehat{\text{CE}}(T)$ of Algorithm~\ref{estimation_algo} asymptotically estimates the long-term causal effect, i.e.,
  $$
  \mathbb{E}(\widehat{\text{CE}}(T))
   = \Ex{R(\alpha_1(T; \ones)) -R(\alpha_0(T; \zeroes))} \equiv \text{CE}(T).
  $$
  %i.e., Algorithm~\ref{estimation_algo} unbiasedly estimates the long-term causal effect of interest in Eq.~\eqref{eq:ce}.
\end{theorem}
\begin{proof}
%% new stuff
Fix a policy $j$ in Algorithm~\ref{estimation_algo} and drop the subscript $j$ in the notation of the algorithm. Therefore we can write:
\begin{align}
\omega & \equiv (\phi_j, \psi_j, B_j)\nonumber\\
\alpha & \equiv \alpha_j(T; j\ones)\nonumber\\
P(\mathcal{D} | \omega) & \equiv P(\mathcal{D}_j | B_j, G_j).\nonumber
\end{align}
The way Algorithm~\ref{estimation_algo} is defined, as the iterations increase the variable $\mu$ is estimating
$$
\lim \mu =\int R(\alpha) P(\mathcal{D}| \omega) p(\alpha, \omega) d\omega d\alpha. 
$$

We now rewrite this integral as follows.
\begin{align}
\lim \mu = \int R(\alpha) P(\mathcal{D}| \omega) p(\alpha, \omega) d\omega d\alpha
 & = \int R(\alpha) P(\mathcal{D}| \alpha, \omega) p(\alpha, \omega) d\omega d\alpha
 \commentEq{ $p(\mathcal{D} | \alpha, \omega) = P(\mathcal{D} | \omega)$}
 \nonumber\\
 & = \int R(\alpha) P(\alpha, \omega|\mathcal{D}) P(\mathcal{D}) d\omega d\alpha
 \commentEq{by Bayes theorem}
 \nonumber\\
 & = P(\mathcal{D}) \int R(\alpha) P(\alpha |\mathcal{D}) d\alpha
 \commentEq{$\omega$ is marginalized out}
 \nonumber\\
 & = P(\mathcal{D}) \Ex{R(\alpha) | \mathcal{D}}.\nonumber
\end{align}

The first equation, $p(\mathcal{D} | \alpha, \omega) = P(\mathcal{D} | \omega)$, holds by definition of the behavioral model: the history of latent behaviors is sufficient for the likelihood of observed actions. Another way to phrase this is that conditional on latent behavior the observed action is independent from any other variable.

Similarly, as the iterations increase the variable $\nu$ is estimating
$$
\lim \nu =\int P(\mathcal{D}| \omega) p(\alpha, \omega) d\omega d\alpha. 
$$
We now rewrite this integral as follows.
\begin{align}
\lim \nu = \int P(\mathcal{D}| \omega) p(\alpha, \omega) d\omega d\alpha
 & = \int P(\mathcal{D}| \alpha, \omega) p(\alpha, \omega) d\omega d\alpha
 \commentEq{because $p(\mathcal{D} | \alpha, \omega) = P(\mathcal{D} | \omega)$}
 \nonumber\\
 & = \int P(\alpha, \omega|\mathcal{D}) P(\mathcal{D}) d\omega d\alpha
 \commentEq{by Bayes theorem}
 \nonumber\\
 & = P(\mathcal{D}) \int P(\alpha |\mathcal{D}) d\alpha
 \nonumber\\
 & = P(\mathcal{D}).\nonumber
\end{align}

By the continuous mapping theorem we conclude that
$$
\lim \mu/\nu \to \Ex{R(\alpha) | \mathcal{D}}.
$$
Thus $\Ex{\lim \mu_1/\nu_1} = \Ex{R(\alpha_1(T; \ones))}$
and $\Ex{\lim \mu_0/\nu_0} = \Ex{R(\alpha_0(T; \zeroes))}$ and so 
$$
\Ex{\lim \mu_1/\nu_1}  - \Ex{\lim \mu_0/\nu_0} \to \Ex{R(\alpha_1(T; \ones))} - \Ex{R(\alpha_0(T; \zeroes))},
$$ 
i.e., Algorithm~\ref{estimation_algo} consistently estimates the long-term causal effect.
\end{proof}

\section{Connection of assumptions to policy invariance}
Assumption~\ref{assumption:ignorability} in our framework is related to
  \emph{policy invariance} assumptions in econometrics of policy
  effects~\citep{heckman05,heckman1998general}. 
  Intuitively, policy invariance posits that given the
  \emph{choice} of policy by an agent, the initial process that 
resulted in this choice does not affect the outcome. For example,
  given that an individual chooses to participate in a tax benefit
  program, the way the individual was assigned to the program
  (e.g., lottery, recommendation, or point of a gun) does not alter
  the outcome that will be observed for that individual.  Our assumption
  is different because we have a temporal evolution of
  population behavior and there is no free choice of an agent about the assignment, since we assume a randomized experiment.
 But our assumption shares the essential aspect of conditional ignorability of assignment that is crucial in causal inference.

\section{Discussion of related methods}
\label{section:related}
Consider the estimand for the Rapoport-Boebel experiment~\citep{rapoport1992mixed}:
\begin{align}
\tau = c^\intercal(\alpha_1(T; \ones) - \alpha_0(T; \zeroes)).\nonumber
\end{align}
Here we discuss how standard methods would estimate
such estimand.  Our goal is to illustrate the
fundamental assumptions underpinning each method, and compare with
our Assumptions~\ref{assumption:initial} and \ref{assumption:ignorability}.
To illustrate we will assume a specific value $c=
(0, 1, 0, 2, 0, 0, 0, 0, 1, 1)^\intercal$.
In discussing these
methods, we will mostly be concerned with how point estimates compare to the true value of the estimand, which here is $\tau = \$0.054$ using the experimental data 
in Table 2.

The naive approach would be to consider only the latest observed time point ($t_0=2$) under the experiment assignment $Z$, and use the observed population actions under $Z$ as an estimate for $\tau$; i.e., 
\begin{align}
\hat{\tau}^{naive}=c^\intercal (\alpha_1(t_0; Z) - \alpha_0(t_0; Z)) = -\$ 0.051.\nonumber
\end{align}
But for this estimate to be unbiased for
$\tau$, we generally require that
$$
\alpha_1(t_0; Z) - \alpha_0(t_0; Z) = \alpha_1(T; \ones) - \alpha_0(T; \zeroes).
$$
The naive estimate therefore makes a direct extrapolation from $t=t_0$ to $t=T$ and from the observed assignment $Z$ to the counterfactual assignments $Z=\ones$ and $Z=\zeroes$. 
This ignores, among other things, the dynamic nature of agent actions.

A more sophisticated approach is to analyze the 
agent actions as a time series. 
For example, \citet{brodersen2013inferring} 
developed a method to estimate the effects of ad campaigns on website visits.  Their method 
was based on the idea of ``synthetic controls'', i.e., 
they created a time-series using different sources of information that would act as the counterfactual 
to the observed time-series after the intervention. 
However, their problem is macroeconometric 
and they work with observational data. 
Thus, there is neither experimental randomized assignment to games, nor strategic interference between agents, nor dynamic agent actions. More crucially, they do not 
study long-term equilbrium effects. 
By construction, in our problem we can leverage behavioral game theory to make more 
informed predictions of counterfactuals to time points after the intervention
at which the distribution of outcomes has stabilized.

Another approach, common in econometrics, is the
{\em difference-in-differences} (DID) estimator~\citep{card1993minimum,
  donald2007inference, ostrovsky2011}.  In our case, this method is
not perfectly applicable because there are no observations before the
intervention, but we can still entertain the idea by considering
period $t=1$ as the pre-intervention period.
The DID estimator compares the difference in outcomes 
before and after the intervention for both the treated
and control groups. In our application, this estimator takes the value
\begin{align}
\label{eq:estimator_did}
\hat{\tau}^{did} = \underbrace{c^\intercal (\alpha_1(t_0; Z) - 
\alpha_1(1; Z))}_{
\text{change in revenue for game 2}} - 
\underbrace{c^\intercal (\alpha_0(t_0; Z) - 
\alpha_0(1; Z))}_{\text{change in revenue for game 1}}
 = -\$0.164.
\end{align}
This estimate is also far from the true value similar to the naive estimate. The DID estimator is unbiased for $\tau$ only if 
there is an additive structure in the actions~\citep{abadie2005semiparametric}, \citep{angrist2008mostly} (Section 5.2), e.g., 
$\alpha_j(t; Z) = \mu_j  + \lambda_t + \epsilon_{jt}$,
where $\mu_j$ is a policy-specific parameter, 
$\lambda_t$ is a temporal parameter, and $\epsilon$ is noise.
The DID estimator thus captures a linear trend in the data by assuming a common parameter for both 
treatment arms ($\lambda_t$) that is canceled out in subtraction in Eq. \eqref{eq:estimator_did}.
The extent to which an additivity assumption 
is reasonable depends on the application, 
however, by definition, it implies ignorability of 
the assignment (i.e., $Z$ does not appear 
in the model of $a_j(t; Z)$), 
and thus it relies on assumptions that are stronger than our assumptions~\citep{abadie2005semiparametric, angrist2008mostly}.

%A more sophisticated econometric method is due to 
In a structural approach, \citet{athey2008comparing} studied
the effects of timber auction format (ascending versus
sealed bid) on competition for timber tracts.
They
estimated bidder valuations from observed data in one auction and 
imputed counterfactual bid distributions in the other auction,
under the assumption of equilibrium play in both auctions.
%\footnote{Note that Athey et.al. \cite{athey2008comparing} consider a different situation than ours. We consider a normal-form game i.e., the payoff structure through Table \ref{table:game} is known and fixed, whereas bidders in an auction carry private valuations when entering an auction.} 
This approach makes two critical implicit assumptions that together
are stronger than Assumption~\ref{assumption:ignorability}.  First,
the bidder valuation distribution is assumed to be a \emph{primitive}
that can be used to impute counterfactuals in other treatment
assignments. In other words, the assignment is independent of bidder
values, and thus it is strongly ignorable.
Second, although imputation is performed for potential outcomes in
equilibrium, which captures the notion of long-term effects, inference
is performed under the assumption of equilibrium play in the
\emph{observed} outcomes, and thus temporal dynamic behavior is assumed
away.

Finally, another popular approach to causality is through {\em
  directed acyclical graphs} (DAGs) between the variables of interest
\citep{pearl2000causality}. For example, \citet{bottou2012} studied the causal effects of the machine learning
algorithm that scores online ads in the Bing search engine on the
search engine revenue.  Their approach was to create a full DAG of the
system including variables such as queries, bids, and prices, and made
a Causal Markov assumption for the DAG. This allows to predict
counterfactuals for the revenue under manipulations of the scoring
algorithm, using only observed data generated from the assumed DAG.
However, a key assumption of the DAG approach 
is that the underlying structural equation model is stable under the treatment assignment, and only edges coming from parents of the manipulated variable need to be removed;
as before, assignment is considered strongly ignorable.
As pointed out by \citet{dash2001caveats} this might be
implausible in equilibrium systems. Consider, for example, a system
where $X \to Y \leftarrow Z$, and a manipulation that sets the
distribution of $Y$ independently of $X, Z$. Then after manipulation
the two edges will need to be removed. However, if in an equilibrium
it is required that $Y \approx X Z$, then the two arrows should be
reversed after the manipulation.  Proper causal inference in
equilibrium systems through causal graphs remains an open area without a well-established
methodology~\citep{dash2005restructuring}. 

Finally we note that there exists the concept of Granger causality~\citep{granger1988some}, which remains important in econometrics. The central idea in Granger causality is predictability, in particular the ability of lagged iterates of a time series $x(t)$ to predict future values of the outcome of interest, which in our case is the population action $\alpha_j(t; Z)$. This causality concept does not take into account the randomization from the experimental design, which is key in statistical causal inference.

% \section{Graphical depiction of Algorithm 1}
% \begin{figure*}[h!]
% \centering
% \includegraphics[scale=0.3]{../explain.PNG}
% \caption{
% {\em Graphical depiction of Algorithm 1 for a fixed game $j$;
% {\bf y-axis}: assignments $Z$, \textbf{x-axis}: time; \textbf{blue line}: observed agent actions $\mathcal{D}$; 
% \textbf{dashed lines}: unobserved actions.
% %
% Algorithm 1 follows the path (A)-(B)-(C)-(D)-(E) as follows:
% %
% %
% \textbf{(A)} use data (blue line) to estimate the parameters $(\phi_j, \psi_j)$ of the temporal model for game $j$ (both are independent of $Z$ by Assumptions 1 and 2);
% \textbf{(B)} estimate the aggregate behavior at $t=0$ using $\phi_j$;
% %
% \textbf{(C)} estimate the aggregate behavior before assignment (t=-1) using Assumption 1, the aggregate behavior at $t=0$, and the experimental randomization;
% %
% \textbf{(D)} derive the initial aggregate behavior in the counterfactual assignment $Z=j\ones$ through $\beta_j(0; j\ones)=\beta^{(0)}$;
% \textbf{(E)} estimate the aggregate behavior at $t=T$ from the behaviors at $t=0$ and the estimated parameters in (A). Estimate the aggregate action at $t=T$ from behaviors at $t=T$ and the behavioral model.}}
% \label{figure:algo}
% \end{figure*}

\pagebreak
\section{Application: \citet{rapoport1992mixed} data }

The following tables report the payoff matrix structure (Table~\ref{table:game} used by Rapoport and Boebel) and the observed data (Table~\ref{table:rapoport}), as reported by~\citet{mckelvey95}.
\label{section:data}
\begin{table}[h]
\centering
\caption{Normal-form game in the experiment 
of Rapoport and Boebel (values $L$ and $W$ are specified
as described in the body of the paper)~\citep{rapoport1992mixed}.} 
\label{table:game}
\renewcommand{\arraystretch}{1.2}
\begin{tabular}{cccccc}
   & $a_1'$ & $a_2'$ & $a_3'$ & $a_4'$ & $a_5'$ \\
\hline
$a_1$ & W  & L  & L  & L  & L  \\
$a_2$ & L  & L  & W  & W  & W  \\
$a_3$ & L  & W  & L  & L  & W  \\
$a_4$ & L  & W  & L  & W  & L  \\
$a_5$ & L  & W  & W  & L  &   L
\end{tabular}
\end{table}

\renewcommand{\arraystretch}{1.2}
\begin{table*}[h!t]
\caption{Experimental data of  Rapoport and Boebel \cite{rapoport1992mixed}, as reported by McKelvey and Palfrey~\cite{mckelvey95}. The data includes frequency of actions for the row agent and the 
column agent in the experiment, broken down 
by game and session. Gray color indicates that we assume the data to be long-term and thus we hold them out of data analysis and only use them to measure predictive performance.
{\em (Note: There are five total actions available to every player according to the payoff structure in Table~\ref{table:game}. The frequencies for actions $a_5,a_5'$ can be inferred because $\sum_{i=1}^5 a_i=1$ and $\sum_{i=1}^5 a_i'=1$.)}}
\label{table:rapoport}
\centering
\begin{tabular}{cc|cccc|cccc}
   &  & \multicolumn{4}{c}{row agent} & \multicolumn{4}{c}{column agent}   \\
Game & Period & $a_1$    & $a_2$    & $a_3$    & $a_4$    & $a_1'$ & $a_2'$ & $a_3'$ & $a_4'$    \\
\hline
1    & 1       & 0.308 & 0.307 & 0.113 & 0.120 & 0.350 & 0.218 & 0.202 & 0.092 \\
1    & 2       & 0.293 & 0.272 & 0.162 & 0.100 & 0.333 & 0.177 & 0.190 & 01.40 \\
1    & 3       & 0.273 & 0.350 & 0.103 & 0.123 & 0.353 & 0.133 & 0.258 & 0.102 \\
1    & 4       & \holdout{0.295} & \holdout{0.292} & \holdout{0.113} & \holdout{0.135} & \holdout{0.372} & \holdout{0.192} & \holdout{0.222} & \holdout{0.063} \\
\hline
2    & 1       & 0.258 & 0.367 & 0.105 & 0.143 & 0.332 & 0.115 & 0.245 & 0.140 \\
2    & 2       & 0.290 & 0.347 & 0.118 & 0.110 & 0.355 & 0.198 & 0.208 & 0.108 \\
2    & 3       & 0.355 & 0.313 & 0.082 & 0.100 & 0.355 & 0.215 & 0.187 & 0.110 \\
2    & 4       & \holdout{0.323} & \holdout{0.270}  & \holdout{0.093}  & \holdout{0.105}  & \holdout{0.343}  & \holdout{0.243}  & \holdout{0.168}  &   \holdout{0.107}   
\end{tabular}
\end{table*}

\renewcommand{\b}[1]{#1}
\section{More details on Bayesian computation}
Here we offer more details about the choices in implementing Algorithm~\ref{estimation_algo} in Section 4.1 of the main paper. For convenience we repeat the content of Section 4.1 in the main paper and then expand with our details.

\textbf{Step 1: Model parameters.} 
For simplicity we assume that the models in the two games share common parameters, and thus 
$(\phi_1, \psi_1, \lambda_1) = (\phi_0, \psi_0, \lambda_0)\equiv (\phi, \psi, \lambda)$, where $\lambda$ are the parameters of the behavioral model to be described in Step 8. Having common parameters also acts as regularization and thus helps estimation. We emphasize that this simplification is not necessary as we could have two different set of parameters for each game. It is crucial, however, that the parameters are stable with respect to the treatment assignment because we need to extrapolate from the observed assignment to the counterfactual ones.

\textbf{Step 4: Sampling parameters and initial behaviors} As explained later we assume that there are 3 different behaviors and thus 
$\phi, \psi, \lambda$ are vectors with 3 components. Let $x\sim U(m, M)$ denote that every component 
of $x$ is uniform on $(m, M)$, independently. We choose diffuse priors for our parameters, 
specifically, $\phi\sim \mathrm{U}(0, 10)$, $\psi\sim \mathrm{U}(-5, 5)$, and $\lambda\sim \mathrm{U}(-10, 10)$.
Given $\phi$ we sample the initial behaviors in the two games as $\beta_1(0; Z)\sim \mathrm{Dir}(\phi)$ and $\beta_0(0; Z)\sim\mathrm{Dir}(\phi)$, independently.

Regarding the particular choices of these distributions, we first note that $\phi$ needs to have positive components because it is used as an argument to the Dirichlet distribution. Larger values than 10 could be used for the components of $\phi$ but the implied Dirichlet distributions would not differ significantly than the ones we use in our experiments. Regarding $\lambda$ we note that its components are used in quantities of the form $e^{\lambda[k] u}$ and so it is reasonable to bound them, and the interval $[-5, 5]$ is diffuse enough given the values of $u$ implied by the payoff matrix in Table~\ref{table:game}. Finally the prior for the temporal model parameters, $\psi$, is also diffuse enough. An alternative would be to use a multivariate normal distribution as the prior for $\psi$ but this would not alter the procedure significantly.

\textbf{Steps 5 \& 7: Pivot to counterfactuals.} 
Since we have a completely randomized experiment (A/B test) it holds that $\rho_Z=0.5$ and 
therefore $\beta^{(0)} = 0.5 (\beta_1(0; Z) + \beta_0(0; Z))$.
Now we can pivot to the counterfactual population behaviors under $Z=\ones$ and $Z=\zeroes$ by setting 
$\beta_1(0; \ones) = \beta_0(0; \zeroes) = \beta^{(0)}$.

\textbf{Step 8: Sample counterfactual behavioral history.}
As the temporal model, we adopt
the {\em lag-one vector autoregressive model}, also known as $\mathrm{VAR}(1)$.
We transform\footnote{The map $y = \mathrm{logit}(x)$ is defined as the function $\Delta^m \to \Reals{m-1}$ such that, for vectors 
$y=(y_1, \ldots, y_{m-1})$ and $x=(x_1, \ldots, x_m)$, $\sum_i x_i=1$, and 
$x_1\neq 0$ wlog, indicates that $y_i=\log(x_{i+1}/x_1)$, for $i=1,\ldots, n-1$.} the population behavior into a new variable $w_t = \mathrm{logit}(\beta_1(t; \ones))\in\Reals{2}$ (also do so for $\beta_0(t; \zeroes)$). Such transformation with a unique inverse is necessary because population behaviors are constrained on the simplex, and thus form so-called compositional data~\citep{aitchison1986statistical, grunwald1993time}. The VAR(1) model implies that
\begin{align}
\label{eq:var}
w_t = \psi[1] \ones + \psi[2] w_{t-1} + \psi[3] \epsilon_t,
\end{align}
where $\psi[k]$ is the $k$th component of $\psi$ and $\epsilon_t \sim \mathcal{N}(0, I)$ is i.i.d. standard bivariate normal. Eq.~\eqref{eq:var} is used to sample 
the behavioral history, $B_j$, from $t=0$ to $t=T$, as 
described in Step 8 of Algorithm~\ref{estimation_algo}.
%We note that the VAR model predicts the temporal evolution of counterfactual population behaviors when $Z=\ones$ or $Z=\zeroes$, and not for any $Z$.

Such sampling is straightforward to do. We simply need to sample the random noises $\epsilon_t$ for every $t\in\{0, \ldots, T\}$, and then compute each $w_t$ successively.
Given the sample $\{w_t: t=0, \ldots, T\}$ we can then transform back to calculate the population behaviors
$\beta_1(t; \ones) = \{\mathrm{logit}^{-1}(w_t) : t=0, \ldots, T\}$---for $B_0$ we repeat the same procedure with a new sample of $\epsilon_t$ since the two games share the same temporal model parameters $\psi$.

\textbf{Step 9: Behavioral model.}
Here we rewrite the specifics of the behavioral model with more details.
In \qlk{p} agents possess increasing levels of sophistication.
Following earlier work~\cite{wright2010beyond}, we adopt $p=3$, and
thus consider a behavioral space with three different behaviors
$\Behaviors = \{b_0, b_1, b_2\}$.

Recall that a behavior $\b \in \Behaviors$ represents the distribution
of actions that an agent will play conditional on adopting that behavior. In
\qlk{p}\ such distributions depend on an assumption of \emph{quantal
  response}, which is defined as follows.
Let $u \in \Reals{|\Actions|}$ denote a vector 
such that $u_a$ is the expected utility of an
agent taking action 
$a \in \mathcal{A}$, and let $G_j$ denote the payoff matrix in game $j$
as in Table \ref{table:game}.
If an agent is facing another agent with strategy (distribution over actions) $b$, then $u = G_j b$.
The quantal best-response with parameter $x$ determines the distribution of actions that the agent will take 
facing expected utilities $u$, 
and is defined as 
\begin{align}
\QBR{u}{x} = \mathrm{expit}(x u),\nonumber
\end{align}
where, for a vector $y$ with elements $y_i$, 
$\mathrm{expit}(y)$ is a vector with elements 
$\exp(y_i) / \sum_i \exp(y_i)$.
The parameter $x \ge 0$ is called the \emph{precision} 
of the quantal best-response. If $x$ is very large then the response is closer to the classical Nash best-response, whereas if $x=0$ 
the agent ignores the utilities and randomizes 
among actions. 

Let $\lambda = (\lambda[1], \lambda[2], \lambda[3])$ 
be the precision parameters.
Let $\alpha(b)$ denote the distribution over actions implied for an agent who adopts behavior.
Given $\lambda$ the model \qlk{3} 
calculates $\alpha(b_k)$, for $k=0, 1, 2$, as follows:
\begin{itemize}
\item Agents who adopt $b_0$, termed \emph{level-0} agents,
 have precision $\lambda_0 = 0$, and thus will randomly pick one action from the action space $\mathcal{A}$.
	Thus, 
	$$
	\alpha(b_0) = \QBR{u}{0} =  (1/|\mathcal{A}|) \ones,
	$$
	regardless of the argument $u$.
\item An agent who adopts $b_1$, termed \emph{level-1} agent, has precision $\lambda[1]$ and assumes that is 
playing against a level-0 type agent. Thus,
 the agent is facing a vector of utilities $u_1 = G_j b_0$,
	and so 
	$$
	\alpha(b_1) = \QBR{u_1}{\lambda[1]}.
	$$
\item An agent who adopts $b_2$, termed \emph{level-2} agent,
 has precision $\lambda[3]$ and assumes is playing against a level-1 agent with precision $\lambda[2]$. 
 Thus, it estimates that it is facing 
strategy $\alpha_{(1)2} = \QBR{u_1}{\lambda[2]}$, where $u_1 = G_j b_0$ as above.
	The expected utility vector of the level-2 agent 
	is $u_2 = G_j \alpha_{(1)2}$, and thus 
	$$
    \alpha(b_2) = \QBR{u_2}{\lambda[3]}.
	$$
\end{itemize}

% The behaviors $b_0, b_1, b_2$ depend 
% only on the parameters $\lambda$ and not on the 
% distribution of behaviors $\beta_j(t; Z)$, because in our game 
% an agent plays against only one other agent.
% Now, let $\Pi_j ({\lambda}) = [b_0 \two b_1 \two b_2]$
% be the $|\A| \times 3$ matrix with the 
% \qlk{3} behaviors parametrized by $\lambda$.
% %
% \qlk{3}\ implies that the 
% expected population action is
% \begin{align}
% \ExCond{\alpha_j(t; Z)}{\beta_j(t; Z)} = \Pi_j (\lambda)  \cdot \beta_j(t; Z).
% \end{align}
% Thus for any observed frequency of actions $\alpha$ in a population of $N$ agents we can calculate the likelihood conditional on the population behavior $\beta$ and \ql{3} parameter $\lambda$ through the equation 
% $$
% P(\alpha | \beta, \lambda) = \mathrm{Multinomial}(N\alpha; \Pi_j(\lambda) \beta),
% $$
% where $\mathram{Multi}((n_1, n_K); (p_1, \ldots, p_K)$ is the 
% density function value at point $(n_1, \ldots, n_K)$ of the multinomial distribution with probabilities $(p_1, \ldots, p_K)$.

Given $G_j$ and $\lambda$ we can therefore write down a $5\times 3$ matrix $Q_j = [\alpha(b_0), \alpha(b_1), \alpha(b_2)]$ where the $k$th column is the distribution over actions played by an agent conditional on adopting behavior $b_{k-1}$. 
Conditional on population action $\beta_j(t; Z)$ the expected population action is  $\bar{\alpha}_j(t; Z) = Q_j \beta_j(t; Z)$.  
The population action $\alpha_j(t; Z)$ is distributed as a multinomial with expectation $\bar{\alpha}_j(t; Z)$, and so 
$P(\alpha_j(t; \ones) | \beta_j(t; \ones), G_j) = \mathrm{Multi}(|\Agents| \cdot \alpha_j(t; \ones); \bar{\alpha}_j(t; \ones))$, where $\mathrm{Multi}(n, p)$ is the multinomial density of observations $n=(n_1, \ldots, n_K)$ with expected frequencies $p=(p_1, \ldots, p_K)$. Hence, the full likelihood for observed actions in game $j$ required in Steps 10 and 11 of Algorithm~\ref{estimation_algo} is given by the product
$$
P(\mathcal{D}_j | B_j, \lambda_j, G_j) = \prod_{t=0}^{T-1} \mathrm{Multi}(|\Agents| \cdot \alpha_j(t; j\ones); \bar{\alpha}_j(t; j\ones)).
$$

\end{document}